\documentclass[a4paper,11pt]{amsart}
\usepackage[letterpaper, margin=1in]{geometry}
\usepackage{latexsym}
\usepackage{amssymb}
\usepackage{amsmath}
\usepackage{amsfonts}
\usepackage{color}

            \DeclareFontFamily{U}{wncy}{}
            \DeclareFontShape{U}{wncy}{m}{n}{%
               <5>wncyr5%
               <6>wncyr6%
               <7>wncyr7%
               <8>wncyr8%
               <9>wncyr9%
               <10>wncyr10%
               <11>wncyr10%
               <12>wncyr6%
               <14>wncyr7%
               <17>wncyr8%
               <20>wncyr10%
               <25>wncyr10}{}

\newtheorem{thm}{Theorem}[section]

\newtheorem{lem}[thm]{Lemma}

\newtheorem{cor}[thm]{Corollary}
\newtheorem{prop}[thm]{Proposition}

\theoremstyle{definition}
\newtheorem*{dfn}{Definition}
\newtheorem{remark}{Remark}[section]

\newtheorem{example}{Example}[]

\newcommand{\N}{\mathbb N}
\newcommand{\Z}{\mathbb Z}

\newcommand{\F}{\mathbb F}

\def\al{\alpha}
\def\be{\beta}

\def\s{\sigma}
\def\d{\delta}

\def\la{\lambda}

\begin{document}

\title[Self-Orthogonality, Self-Duality of MPC's]{On Self-Orthogonality and Self-Duality of Matrix-Product Codes over Commutative Rings}

\author{Abdulaziz Deajim}
\address[A. Deajim]{Department of Mathematics, King Khalid University,
P.O. Box 9004, Abha, Saudi Arabia} \email{deajim@kku.edu.sa, deajim@gmail.com}

\author{Mohamed Bouye}
\address[M. Bouye]{Department of Mathematics, King Khalid University,
P.O. Box 9004, Abha, Saudi Arabia} \email{medeni.doc@gmail.com}

\date{\today}

\begin {abstract}
Let $R$ be a commutative ring with identity. The paper studies the problem of self-orthogonality and self-duality matrix-product codes (MPCs) over $R$. Some methods as well as special matrices are introduced for the construction of such MPCs. A characterization of such codes (in a special case) is also given. Some concrete examples are presented throughout the paper.
\end {abstract}
\maketitle

\section{{\bf Introduction}}\label{intro}
Besides being coding-theoretically interesting in their own right, Euclidean self-dual and self-orthogonal codes proved to be very useful in diverse areas of mathematics and its applications such as group theory, combinatorial designs, communication systems, and lattice theory (see \cite{DKL}, \cite{Dzh}, \cite{Pless}, and \cite{Wan}). On the other hand, Blackmore and Norton, in their pioneering paper \cite{BN}, introduced the important notion of matrix-product codes (MPCs) over finite fields. An MPC utilizes a finite list of (input) codes of the same length to produce a longer code. The parameters and decoding capabilities of some of such codes were studied by many authors (see for instance \cite{BN}, \cite{HLR}, and \cite{HR}). Some authors also considered MPCs and some of their properties over certain finite commutative rings (see for instance \cite{As}, \cite{BD}, and \cite{FLL}).

To connect the aforementioned concepts, one would ask: Under what conditions can one construct a self-orthogonal or self-dual MPC over a finite field? To the best of the authors' knowledge, the work of Mankean and Jitman \cite{MJ} (which is a follow-up on \cite{M}) was the first published work that addresses this question. The aim of this paper is to consider the above question over an arbitrary commutative ring with unity (finite or infinite). Among other contributions, we generalize some results of \cite{MJ} and, further, relax some of their requirements.

For the paper to be self-contained, we give in Section \ref{prelim} the necessary preliminary definitions and results. It is assumed throughout the paper that the ring over which the codes are considered is a commutative ring with unity. In Section \ref{ss}, sufficient conditions are given for an MPC to be self-orthogonal (Theorems \ref{self-orth 1} and \ref{self-orth 2}, and Corollary \ref{orthog matrix}) or self-dual (Theorem \ref{self-dual}). Theorem \ref{self MPC} introduces a condition under which we get a characterization of self-orthogonal and self-dual MPCs. Theorem \ref{dual of MPC} gives a description of the dual of an MPC as an MPC, generalizing what is known over finite fields \cite{BN}, finite chain rings \cite{As}, and finite commutative rings \cite{BD}. In Section \ref{application}, special matrices are introduced to be used in constructing self-orthogonal and self-dual MPCs with enhanced minimum distances. Some concrete examples are also given throughout the paper.

\section{{\bf Preliminaries}} \label{prelim}

Throughout this paper $R$ denotes a commutative ring with identity $1$ and $U(R)$ its multiplicative group of units. To present our results under possibly broad assumptions, we choose not to put further restrictions on $R$ unless they are really needed. Recall that {\it a code $C$ over $R$ of length $m$} is a subset of $R^m$; while such a code is said to be {\it linear over $R$} if it is an $R$-submodule of $R^m$. A linear code $C$ over $R$ is said to be {\it free over $R$} if it is free as an $R$-module, where the cardinality of a (free) $R$-basis of $C$ is called the {\it rank of $C$}. If $C$ is a free linear code over $R$ of length $m$ and rank $r$, then a matrix $G\in M_{r\times m}(R)$ whose rows form an $R$-basis of $C$ is called a {\it generating matrix} of $C$. In this case, a given element of $C$ is precisely of the form $xG$ for a unique $x\in R^r$.

Consider the Euclidean inner product on $R^m$ defined by $\langle x , y\rangle = x_1 y_1 +\dots +x_m y_m$ for $x=(x_1, \dots, x_m)$ and $y=(y_1, \dots, y_m)$. If $C$ is a linear code over $R$ of length $m$, define {\it the dual code $C^\perp$ of $C$} to be
$$C^\perp=\{x\in R^m\,|\, \mbox{$\langle x, c\rangle\,=0$ for all $c\in C$}\}.$$
It is easily checked that $C^\perp$ is a linear code over $R$ as well. A linear code $C$ over $R$ is said to be {\it self-orthogonal} if $C\subseteq C^\perp$, and {\it self-dual} if $C=C^\perp$.

If $C$ is a linear code over $R$ of length $m$, recall that the Hamming distance on $C$ is defined by
$$d(x,y)= |\,\{i \,|\, x_i \neq y_i\}\,|$$
for $x=(x_1, \dots, x_m), y=(y_1, \dots, y_m)\in C$. Any distance in this paper is to mean the Hamming distance. The minimum distance of $C$ is then defined to be
$$d(C)=\min\{d(x,y)\,|\, x,y\in C, x\neq y\}.$$
The Hamming weight is defined on $C$ by $\mbox{wt}(x)=d(x,0)$ for $x\in C$. So, $\mbox{wt}(x)=|\,\{i \,|\, x_i \neq 0\}\,|$ for $x=(x_1, \dots, x_m)\in C$. It can be checked that $d(C)=\min\{\mbox{wt}(x)\,|\, x\in C, x\neq 0\}$.

For positive integers $s$ and $l$, we denote by $M_{s\times l}(R)$ the set of all $s\times l$ matrices with entries in $R$. In this paper, we always assume that $s\leq l$. For $A\in M_{s\times l}(R)$, denote by $A^t$ the usual transpose of $A$. If the rows of $A\in M_{s\times l}(R)$ are linearly independent over $R$, we say that {\it $A$ has full rank}. For $\la_1, \dots, \la_s\in R$, denote by $\mbox{diag}(\la_1, \dots, \la_s)\in M_{s\times s}(R)$ the diagonal matrix whose entry in position $i,i$ is $\la_i$, and denote by $\mbox{adiag}(\la_1, \dots, \la_s)\in M_{s\times s}(R)$ the anti-diagonal matrix whose entry in position $i,(s-i+1)$ is $\la_i$. A matrix $A\in M_{s\times s}(R)$ is {\it non-singular} or {\it invertible} if and only if $\mbox{det}(A)\in U(R)$. Note that if $A\in M_{s\times s}(R)$ and $AA^t=\mbox{diag}(\la_1, \dots, \la_s)$ or $\mbox{adiag}(\la_1, \dots, \la_s)$ with $\la_i\in U(R)$ for $i=1, \dots, s$, then both $A$ and $A^t$ are non-singular, as classical properties of the determinant remain valid over commutative rings (see \cite[I.D]{Mc}).

Let $C_1, \dots, C_s$ be linear codes over $R$ of length $m$ and $A\in M_{s\times l}(R)$. Denote by $[C_1 \dots C_s]\,A$ {\it the matrix-product code} (MPC) over $R$ of length $ml$ in the sense of \cite{BN} (see also \cite{As} and \cite{FLL}); that is
$$[C_1 \dots C_s]\,A=\{(c_1 \dots c_s)A\,| \, c_i\in C_i,\, 1\leq i \leq s\},$$
where $(c_1 \dots c_s)$ is an $m\times s$ matrix whose $i$th column is $c_i\in C_i$ written in column form. The codes $C_1, \dots, C_s$ are called the {\it input codes of $[C_1 \dots C_s]\,A$}. Note that as $C_1, \dots, C_s$ are linear over $R$, so is $[C_1 \dots C_s]\,A$. If $A=(a_{i,j})$, a typical codeword $c$ of $[C_1 \dots C_s]\,A$ is an $m\times l$ matrix $c=(x_{i,j})$ with $x_{i,j}=\sum_{i=1}^s a_{i,j}c_{k,i}$, where $c_{k,i}$ is the $k$th component of $c_i$. On the other hand, as the $j$th column of $c$ is $\sum_{i=1}^s a_{i,j}c_i$, we can also look at $c$ in the form
$$c=\left(\sum_{i=1}^s a_{i,1}c_i, \sum_{i=1}^s a_{i,2} c_i, \dots, \sum_{i=1}^s a_{i,l}c_i\right).$$
If $I_s\in M_{s\times s}(R)$ is the identity matrix, we denote the matrix-product code $[C_1 \dots C_s]\,I_s$ by $[C_1 \dots C_s]$. Note that codewords in $[C_1 \dots C_s]$ are matrices of the form $(c_1 \dots c_s)$ with $c_i\in C_i$ for $i=1, \dots, s$.

If $A=(a_{i,j})\in M_{s\times l}(R)$ is of full rank and $C_i$ is a free linear code over $R$ of length $m$, rank $r_i$, and a generating matrix $G_i\in M_{r_i \times m}(R)$ for $i=1, \dots, s$, respectively, it was shown in \cite{BD} that $[C_1 \dots C_s]\,A$ is free of rank $r=\sum_{i=1}^s r_i$ with a generating matrix $(a_{i,j}G_i)\in M_{r\times lm}(R)$.


\section{{\bf Self-Orthogonal and Self-Dual Matrix-Product Codes}}\label{ss}

The following two theorems give sufficient conditions for an MPC to be self-orthogonal.

\begin{thm}\label{self-orth 1}
Let $A=(a_{i,j})\in M_{s\times l}(R)$ be such that $AA^t=\mbox{diag}(\la_1, \dots, \la_s)$ for some $\la_1, \dots, \la_s \in R$. Suppose that $C_1, \dots, C_s$ are linear codes over $R$ of the same length such that, for $i=1, \dots, s$, $C_i$ is self-orthogonal whenever $\la_i\neq 0$. Then, $[C_1 \dots C_s]\,A$ is self-orthogonal.
\end{thm}

\begin{proof}
Let $c\in [C_1 \dots C_s]\,A$. In order to show that $c\in ([C_1 \dots C_s]\,A)^\perp$, we prove that $\langle c, c'\rangle =0$ for any $c'\in [C_1 \dots C_s]\,A$. Let $$\mbox{$c=(\sum_{i=1}^s a_{i,1}c_i, \sum_{i=1}^s a_{i,2} c_i, \dots, \sum_{i=1}^s a_{i,l}c_i)$ and $c'=(\sum_{i=1}^s a_{i,1}c_i', \sum_{i=1}^s a_{i,2} c_i', \dots, \sum_{i=1}^s a_{i,l}c_i')$}$$ for $c_i, c_i'\in C_i$, $i=1, \dots s$. Then we have
\begin{align*}
\langle c, c'\rangle  &= \sum_{i=1}^s\sum_{j=1}^s a_{i,1}a_{j,1}\; \langle c_i, c_j'\rangle + \sum_{i=1}^s\sum_{j=1}^s a_{i,2}a_{j,2}\; \langle c_i, c_j'\rangle+ \dots +\sum_{i=1}^s\sum_{j=1}^s a_{i,l}a_{j,l}\; \langle c_i, c_j'\rangle\\
&= (\sum_{j=1}^s a_{1,j}a_{1,j}) \;\langle c_1, c_1'\rangle + \dots + (\sum_{j=1}^s a_{1,j}a_{s,j}) \;\langle c_1, c_s'\rangle\\
& \qquad + (\sum_{j=1}^s a_{2,j}a_{1,j}) \;\langle c_2, c_1'\rangle + \dots + (\sum_{j=1}^s a_{2,j}a_{s,j}) \;\langle c_2, c_s'\rangle\\
&\qquad\qquad +\cdots \\
&\qquad\qquad\qquad + (\sum_{j=1}^s a_{s,j}a_{1,j}) \;\langle c_s, c_1'\rangle + \dots + (\sum_{j=1}^s a_{s,j}a_{s,j}) \;\langle c_s, c_s'\rangle.
\end{align*}
Now, for each $i=1, \dots s$, $(\sum_{j=1}^s a_{i,j}a_{j,i}) \langle c_i, c_i'\rangle =\la_i \langle c_i, c_i'\rangle =0$, because either $\la_i=0$ or, else, $\langle c_i,  c_i'\rangle =0 $ (since $C_i$ is self-orthogonal in this case). On the other hand, for $i\neq k$, $(\sum_{j=1}^s a_{i,j}a_{k,j}) \langle c_i, c_k'\rangle =0$ as well, because $\sum_{j=1}^s a_{i,j}a_{k,j}$ is the entry of $AA^t$ in position $i,k$, which is 0 by assumption. Hence, $\langle c, c'\rangle =0$ as desired.
\end{proof}

\begin{remark}
Theorem \ref{self-orth 1} generalizes, and relaxes the assumptions of, \cite[Theorem III.1]{MJ} which assumes that $R$ is a finite field, the input codes are free and are all self-orthogonal.
\end{remark}

\begin{example}\label{ex 1}
Over $\Z_{20}$, let $C_1=10Z_{20}=\{0, 10\}$ and $C_2=4\Z_{20}=\{0,4,8,12,16\}$. It can be seen that $C_1^\perp=2\Z_{20}=\{0,2,4,6,8,10,12,14,16,18\}$ and, thus, $C_1$ is self-orthogonal. Note, on the other hand, that $C_2^\perp=5\Z_{20}=\{0,5,10,15\}$ and, thus, $C_2$ is not self-orthogonal. Take $A=\left(\begin{array}{cc} 1 &2 \\ 0 &0 \end{array}\right)$. So, $AA^t=\mbox{diag}(5,0)$. By Theorem \ref{self-orth 1}, $[C_1 C_2]\,A$ is self-orthogonal. Indeed, $[C_1C_2]\,A=10\Z_{20}\times \{0\}$, $([C_1 C_2]\,A)^\perp =2\Z_{20}\times \Z_{20}$, and obviously $[C_1 C_2]\,A\subseteq ([C_1 C_2]\,A)^\perp$.
\end{example}

\begin{thm}\label{self-orth 2}
Let $A\in M_{s\times l}(R)$ be such that $AA^t=\mbox{adiag}(\la_1, \dots, \la_s)$ for some $\la_1, \dots, \la_s \in R$. Suppose that $C_1, \dots, C_s$ are linear codes over $R$ of the same length such that, for $i=1, \dots, s$, $C_i\subseteq C_{s-i+1}^\perp $ whenever $\la_i\neq 0$. Then, $[C_1 \; \dots \; C_s]\,A$ is self-orthogonal.
\end{thm}

\begin{proof}
Similar to the proof of Theorem \ref{self-orth 1} with the obvious adjustments.
\end{proof}

\begin{remark}
Theorem \ref{self-orth 2} generalizes, and relaxes the assumptions of, \cite[Theorem III.4]{MJ} which assumes that $R$ is a finite field, the input codes are free, and $C_i\subseteq C_{s-i+1}^\perp $ for all $i=1, \dots, s$.
\end{remark}

\begin{example}\label{ex 2}
Let $C_1$ and $C_2$ be as in Example \ref{ex 1}. Then, $C_1\subseteq C_2^\perp$ and $C_2\subseteq C_1^\perp$. Let $B=\left(\begin{array}{cccc} 0&2&0&4\\ 0&4&2&0\end{array}\right)$. Then, $BB^t=\mbox{adiag}(8,8)$. It then follows from Theorem \ref{self-orth 2} that both $[C_1 C_2]\,B$ and $[C_2 C_1]\,B$ are self-orthogonal. Indeed, $$[C_1 C_2]\,B=\{(0,0,0,0),(0,16,8,0),(0,12,16,0),(0,8,4,0),(0,4,12,0)\},$$
$$[C_2 C_1]\,B=\{(0,0,0,0), (0,8,0,16), (0,16,0,12), (0,4,0,8), (0,12,0,4)\},$$
and it can be checked that $\langle(a,b,c,d), (a',b',c',d')\rangle  =0$ and $\langle(e,f,g,h), (e',f',g',h')\rangle =0$ for all $(a,b,c,d),(a',b',c',d')\in [C_1 C_2]\,B$, and $(e,f,g,h), (e',f',g',h')\in [C_2 C_1]\,B$. Thus, $[C_1 C_2]\,B\subseteq ([C_1 C_2]\,B)^\perp$ and $[C_2 C_1]\,B\subseteq ([C_2 C_1]\,B)^\perp$.
\end{example}

The equality $([C_1 \dots C_s]\,A)^\perp=[C_1^\perp \dots C_s^\perp]\,(A^{-1})^t$ is well-known to hold if $R$ is a finite field or a finite chain ring, $C_i$ are free over $R$, and $A\in M_{s\times s}(R)$ is non-singular (see \cite{BN} and \cite{As} for instance). In \cite{BD}, this equality was shown to hold over any finite commutative ring. In Theorem \ref{dual of MPC} below, we show that this fact remains true over any commutative ring $R$ without even assuming that the input codes are free over $R$.

\begin{thm}\label{dual of MPC}
Let $A\in M_{s\times s}(R)$ be non-singular and $C_1, \dots, C_s$ linear codes of length $m$ over $R$. Then, the dual of the matrix product code $[C_1 \dots C_s]\,A$ is given by $$([C_1, \dots, C_s]\,A)^\perp=[C_1^\perp \dots C_s^\perp]\,(A^{-1})^t.$$
\end{thm}

\begin{proof}
Let $A=(a_{i,j})$. We first show that $([C_1 \dots C_s]\,A)^\perp \subseteq [C_1^\perp \dots C_s^\perp]\,(A^{-1})^t$. Let $x=(x_1, \dots, x_s)\in ([C_1 \dots C_s]\,A)^\perp$. Note that $x_i\in R^m$ for every $i$. Then, $\langle x, c\rangle  =0$ for every $c\in [C_1 \dots C_s]\,A$. Then we have, for every $j=1, \dots, s$ and every $c_j\in C_j$,
\begin{align*}
0&= \langle(x_1, \dots, x_s), (\sum_{j=1}^s a_{j,1}c_j, \dots, \sum_{j=1}^s a_{j,s}c_j)\rangle\\
 &= \sum_{i=1}^s \langle x_i, \sum_{j=1}^s a_{j,i}c_j\rangle.
\end{align*}
For a fixed $j$, apply the above equality to all codewords of $[C_1 \dots C_s]\,A$ of the form $(c_1, \dots,c_j, \dots, c_s)\,A$ with $c_i=0$ for $i\neq j$ and $c_j$ running over all codewords of $C_j$ to get $$0=\sum_{i=1}^s \langle x_i, a_{j,i}c_j\rangle  =\sum_{i=1}^s \langle a_{j,i}x_i, c_j\rangle =\langle\sum_{i=1}^s a_{j,i}x_i, c_j\rangle.$$
It then follows that $\sum_{i=1}^s a_{j,i}x_i \in C_j^\perp$. Doing this for every $j=1, \dots, s$, we get $(x_1, \dots, x_s)\,A^t\in [C_1^\perp, \dots, C_s^\perp]$, which yields that $x=(x_1, \dots, x_s)\in [C_1^\perp, \dots, C_s^\perp]\,(A^{-1})^t$.

Conversely, we show that $[C_1^\perp \dots C_s^\perp]\,(A^{-1})^t \subseteq ([C_1 \dots C_s]\,A)^\perp $. Let $x\in [C_1^\perp \dots C_s^\perp]\,(A^{-1})^t$. Then $x=(x_1, \dots, x_s)=(c_1^\perp, \dots, c_s^\perp)\,(A^{-1})^t$, where $c_i^\perp\in C_i^\perp$ for $i=1, \dots, s$. It follows that $(x_1, \dots, x_s)\,A^t=(c_1^\perp, \dots, c_s^\perp)$ and, thus, $\sum_{i=1}^s a_{j,i}x_i =c_j^\perp \in C_j^\perp$ for every $j=1, \dots, s$. This means that, for any fixed $j$ and all $y_j\in C_j$,
$$0=\langle\sum_{i=1}^s a_{j,i}x_i, y_j\rangle  = \sum_{i=1}^s  \langle a_{j,i}x_i, y_j\rangle  = \sum_{i=1}^s \langle x_i, a_{j,i}y_j\rangle.$$
Doing this process for every $j=1, \dots, s$ yields that $\sum_{i=1}^s \langle x_i, \sum_{j=1}^s a_{j,i}y_j\rangle =0$ for all $y_j\in C_j$. So,
$$0= \langle(x_1, \dots, x_s), (\sum_{j=1}^s a_{j,1}y_j, \dots, \sum_{j=1}^s a_{j,s}y_j)\rangle$$ for all $y_j\in C_j$, $j=1, \dots, s$. Thus, $\langle x, c\rangle =0$ for every $c\in [C_1 \dots C_s]\,A$ and, hence, $x\in ([C_1 \dots C_s]\,A)^\perp$.
\end{proof}

\begin{cor}\label{orthog matrix}
Keep the assumptions of Theorem \ref{dual of MPC}, and assume further that $A$ is orthogonal (i.e. $A=(A^{-1})^t$). Then,

1. $([C_1 \dots C_s]\,A)^\perp=[C_1^\perp \dots C_s^\perp]\,A$.

2. If $C_i$ is self-orthogonal for each $i=1, \dots, s$, then so is $[C_1 \dots C_s]\,A$.

3. If $C_i$ is self-dual for each $i=1, \dots, s$, then so is $[C_1 \dots C_s]\,A$.

4. If $C_i^\perp\subseteq C_i$ for each $i=1, \dots, s$, then $([C_1 \dots C_s]\,A)^\perp \subseteq [C_1 \dots C_s]\,A$.
\end{cor}

\begin{proof}
Clear!
\end{proof}

\begin{remark}
For part 4 of Corollary \ref{orthog matrix} to hold, orthogonality of $A$ is sufficient but not necessary (see \cite[Theorem 13]{GHR}).
\end{remark}

Note that Corollary \ref{orthog matrix} gives, in particular, a sufficient condition for the self-duality of an MPC. The following theorem gives another sufficient condition.

\begin{thm}\label{self-dual}
Let $A\in M_{s\times s}(R)$ be such that $AA^t=\mbox{adiag}(\la_1, \dots, \la_s)$ for $\la_1, \dots, \la_s \in U(R)$. Suppose that $C_1, \dots, C_s$ are linear codes of the same length over $R$ such that $C_i=C_{s-i+1}^\perp $ for $i=1, \dots, s$. Then, $[C_1 \dots C_s]\,A$ is self-dual.
\end{thm}

\begin{proof}
Let $A=(a_{i,j})$. The containment $[C_1 \dots C_s]\,A \subseteq ([C_1 \dots C_s]\,A)^\perp$ follows from Theorem \ref{self-orth 2}. It remains to show that $([C_1 \dots C_s]\,A)^\perp \subseteq [C_1 \dots C_s]\,A$. Let $x\in ([C_1 \dots C_s]\,A)^\perp$. Then, by Theorem \ref{dual of MPC}, $x=[c_1', c_2', \dots, c_s'](A^{-1})^t$ for some $c_i'\in C_i^\perp$, $i=1, \dots, s$. As $C_i=C_{s-i+1}^\perp $ for each $i=1, \dots, s$, $C_{s-i+1}=C_{s-(s-i+1)+1}^\perp=C_i^\perp$ for each $i=1, \dots, s$. Thus, $c_i' \in C_{s-i+1}$ for each $i=1,\dots, s$. Let $\la_i'\in R$ be such that $\la_i\la_i'=1$ and set $e_{s-i+1}=\la_i' c_i'$ for $i=1, \dots, s$. It follows that $e_{s-i+1} \in C_{s-i+1}$ for $i =1, \dots, s$ since $\la_i' \in R$ and $C_{s-i+1}$ is linear over $R$. As $AA^t=\mbox{adiag}(\la_1, \dots, \la_s)$, it follows that
$$(A^{-1})^t= (\la_i' a_{s-i+1,j})=  \left( \begin{array}{cccc} \la_1' a_{s,1} & \la_1' a_{s,2}  & \dots & \la_1' a_{s,s}\\ \la_2' a_{s-1,1} & \la_2' a_{s-1,2}  & \dots & \la_2' a_{s-1,s}\\ \vdots & \vdots & \dots & \vdots \\  \la_s' a_{1,1} & \la_s' a_{1,2}  & \dots & \la_s' a_{1,s} \end{array} \right).$$
So,
\begin{align*}
x &= [c_1', c_2', \dots, c_s'](A^{-1})^t\\
  &= \left(\sum_{i=1}^s \la_i' a_{s-i+1,1} c_i', \sum_{i=1}^s \la_i' a_{s-i+1, 2}c_i', \dots , \sum_{i=1}^s \la_i' a_{s-i+1, s} c_i'\right)\\
  &= \left(\sum_{i=1}^s a_{i, 1} e_i, \sum_{i=1}^s a_{i, 2} e_i, \dots, \sum_{i=1}^s a_{i, s} e_i \right)\\
  &= [e_1, e_2, \dots, e_s]A \in [C_1 \dots C_s]\,A.
\end{align*}
\end{proof}

\begin{remark}
Theorem \ref{self-dual} generalizes, and relaxes the assumptions of, \cite[Corollary III.6]{MJ} which assumes that $R$ is a finite field, the input codes are free, and $\la_1=\dots =\la_s$.
\end{remark}

\begin{example}
Over $\Z_{25}$, let $C=\Z_{25}(1,7)$. Then $C$ is a linear self-dual code of length 2 over $\Z_{25}$. Indeed, for $a,b\in Z_{25}$, $\langle(a,7a), (b,7b)\rangle =ab(1+49)=0$. So, $C\subseteq C^\perp$. On the other hand, for $(x,y)\in C^\perp$ and $(a,7a)\in C$, $\langle(x,y), (a,7a)\rangle =0$ implies that $a(x+7y)=0$. Taking $a\in U(\Z_{25})$ yields $x+7y=0$ and, thus, $y=-7^{-1}x=7x$. So, $(x,y)=(x,7x)\in C$. Thus, $C^\perp \subseteq C$. Now, take $A=\left(\begin{array}{cc} 1&7\\7&1 \end{array}\right)$. Then $AA^t=\mbox{adiag}(14,14)$ and $14\in U(\Z_{25})$. It then follows from Theorem \ref{self-dual} that $[C \, C]\,A$ is self-dual. As a side, it can be checked that $[C\,C]\,A$ contains no codeword of weight, while it contains, for instance, the codeword $\left(\begin{array}{cc} 14 & 0 \\ 23 & 0 \end{array}\right)$, which is of weight 2. So, the minimum distance of this MPC is 2, which is the same as the minimum distance of $C$. On the other hand, $C$ is free of rank 1, so its information rate is $1/2$. Similarly, $[C\,C]\,A$ is free of rank 2 and length 4, so its information rate is also $1/2$. So, despite the fact that this MPC caused doubling of the length of $C$ and its cardinality, it nonetheless preserved the self-duality and both the minimum distance and the information rate of $C$.
\end{example}

Our next goal is Theorem \ref{self MPC}, in which we give a sufficient condition for the equivalence of self-orthogonality (resp. self-duality) of an MPC and self-orthogonality (resp. self-duality) of its input codes.

\begin{lem}\label{C_A}
Let $A=(a_{i,j})\in M_{s\times s}(R)$ be non-singular and $C_1, \dots, C_s$ linear codes of the same length over $R$. Then $[C_1 \dots C_s]\,A=[C_1 \dots C_s]$ if either of the following holds:

1. $C_1 \subseteq C_2 \subseteq \dots \subseteq C_s$ and $A$ is upper triangular,

2. $C_s \subseteq C_{s-1} \subseteq \dots \subseteq C_1$ and $A$ is lower triangular,

3. $A$ is diagonal, or

4. $C_1 = C_2 = \dots = C_s$.
\end{lem}

\begin{proof} \hfill

1. Suppose that $C_1 \subseteq C_2 \subseteq \dots \subseteq C_s$ and $A$ is upper triangular. Then $a_{i,j}=0$ for $i> j$. Moreover, $a_{j,j}\in U(R)$ for all $j=1, \dots, s$ since $A$ is non-singular. It follows that
$$[C_1 \dots C_s]\,A=[a_{1,1}C_1, a_{1,2}C_1+a_{2,2}C_2, \dots, a_{1,s}C_1+a_{2,s}C_2 +\dots + a_{s,s}C_s].$$
Since $a_{1,1}\in U(R)$ and $C_1$ is linear, $a_{1,1}C_1=C_1$. Similarly, $a_{2,2}C_2=C_2$. Since $C_1 \subseteq C_2$ and $C_2$ is linear, $a_{1,2}C_1\subseteq C_2$. It follows that $a_{1,2}C_1 + a_{2,2}C_2=a_{1,2}C_1 + C_2 =C_2$. We continue in this manner to get that $a_{1,j}C_1 +a_{2,j}C_2 + \dots + a_{j,j}C_j = C_j$ for all $j=1 ,\dots, s$. Thus, $[C_1 \dots C_s]\,A=[C_1 \dots C_s]$ as claimed.

2. If $C_s \subseteq C_{s-1} \subseteq \dots \subseteq C_1$ and $A$ is lower triangular, the proof is similar to case 1 above with the obvious adjustments.

3. Suppose that $A$ is diagonal. So, $a_{i,j}=0$, for all $i\neq j$, and $a_{j,j}\in U(R)$, for all $j=1, \dots, s$ (since $A$ is non-singular). It follows that $$[C_1 \dots C_s]\,A= [a_{1,1}C_1 \dots a_{s,s}C_s]=[C_1 \dots C_s]$$
because $a_{j,j}C_j=C_j$, as $a_{j,j}\in U(R)$ and $C_j$ is linear for every $j=1, \dots, s$.

4. Let $x\in [C \dots C]\,A$. So, $x=(c_1 \dots c_s)\,A$ for some $c_1, \dots, c_s\in C$. By definition, $x=(\sum_{i=1}^s a_{i,1}c_i, \dots, \sum_{i=1}^s a_{i,s}c_i)$. As $\sum_{i=1}^s a_{i,j}c_i\in C$ for all $j=1, \dots, s$, $x\in [C \dots C]$ and, thus, $[C \dots C]\, A\subseteq [C \dots C]$. Conversely, let $x\in [C \dots C]$. Applying the previous argument to $A^{-1}$, we have $[C \dots C]\,A^{-1} \subseteq [C \dots C]$. Now, $xA^{-1}\in [C \dots C]\,A^{-1} \subseteq [C \dots C]$. Hence, $x\in [C \dots C]\,A$ and, therefore, $[C \dots C]\subseteq [C \dots C]\,A$.
\end{proof}

\begin{thm}\label{self MPC}
Let $A\in M_{s\times s}(R)$ be non-singular and $C_1, \dots, C_s$ linear codes of the same length over $R$ such that $[C_1 \dots C_s]\,A=[C_1 \dots C_s]$. Then,

1. $[C_1 \dots C_s]\,A$ is self-orthogonal if and only if $C_1, \dots, C_s$ are all self-orthogonal.

2. $[C_1 \dots C_s]\,A$ is self-dual if and only if $C_1, \dots, C_s$ are all self-dual.
\end{thm}

\begin{proof} Assume that $[C_1 \dots C_s]\,A=[C_1 \dots C_s]$. Note that $[C_1 \dots C_s]=[C_1 \dots C_s]\,I_s$. By Theorem \ref{dual of MPC}, we have
$$([C_1 \dots C_s]\,A)^\perp=([C_1 \dots C_s]\,I_s)^\perp =[C_1^\perp \dots C_s^\perp]\,(I_s^{-1})^t=[C_1^\perp \dots C_s^\perp].$$
So, $[C_1 \dots C_s]\,A$ is self-orthogonal (resp. self-dual) if and only if $[C_1 \dots C_s]\subseteq [C_1^\perp \dots C_s^\perp]$ (resp. $[C_1 \dots C_s]=[C_1^\perp \dots C_s^\perp]$). The claimed conclusion is now obvious.
\end{proof}


\begin{cor}
Let $A\in M_{s\times s}(R)$ be non-singular and $C_1, \dots, C_s$ linear codes of the same length over $R$ such that any of the conditions of Lemma \ref{C_A} holds. Then,

1. $[C_1 \dots C_s]\,A$ is self-orthogonal if and only if $C_1, C_2, \dots, C_s$ are all self-orthogonal.

2. $[C_1 \dots C_s]\,A$ is self-dual if and only if $C_1, C_2, \dots, C_s$ are all self-dual.
\end{cor}

\begin{proof}
Apply Lemma \ref{C_A} and Theorem \ref{self MPC}.
\end{proof}

\section{\bf Applications}\label{application}



For a full-rank matrix $A\in M_{s\times l}(R)$, denote by $C_{R_i}$ the code of length $l$ over $R$ generated by the upper $i$ rows of $A$ for $i=1, \dots, s$. For linear codes $C_1, \dots, C_s$ of the same length over $R$ with minimum distances $d_1, \dots, d_s$, respectively, it was recently shown in \cite{BD} that the minimum distance $d$ of the matrix-product code $[C_1 \dots C_s]\,A$ satisfies: \begin{equation} d\geq \mbox{min}\{d_i \d_i\}_{1\leq i\leq s},\end{equation} where $\d_1, \dots, \d_s$ are the minimum distances of $C_{R_1}, \dots, C_{R_s}$, respectively. This result was first known over finite fields (\cite{BN}) and finite chain rings (\cite{As}), and was lately shown to hold over any commutative ring with identity (\cite{BD}).

\begin{lem}\label{diag1}
Let $R$ be such that $2=2\cdot 1_R$ is not a zero divisor, and let $A=\left(\begin{array}{ccc} 1&u&1\\-1&0&1\end{array}\right)$ with $u\in R$ not a zero divisor. Then, $AA^t=\mbox{diag}(2+u^2,2)$, $\d_1=3$, and $\d_2=2$.
\end{lem}

\begin{proof}
It is straightforward to check that $AA^t=\mbox{diag}(2+u^2,2)$. As $C_{R_1}=R(1,u,1)$, an element of $C_{R_1}$ is of the form $(\al, \al u, \al)$ for some $\al\in R$. Suppose that $\mbox{wt}(\al,\al u, \al)=1$. It is clearly impossible to have this assumption with $\al\neq 0$. But if $\al=0$, then $(\al,\al u, \al)=(0,0,0)$, which is impossible as well. So, there is no $\al\in R$ such that $\mbox{wt}(\al,\al u, \al)=1$. Similarly, suppose that $\mbox{wt}(\al, \al u, \al)=2$. It is obvious that $\al$ cannot be zero. But if $\al\neq 0$, then we must have $\al u=0$. Since $u$ is not a zero divisor, $\al=0$, a contradiction. So, there is no $\al\in R$ such that $\mbox{wt}(\al, \al u, \al)=2$. Thus, $\d_1=3$.

On the other hand, as $C_{R_2}=R(1,u,1)+R(-1, 0, 1)$, an element of $C_{R_2}$ is of the form \linebreak $(\al-\be, \al u, \al+\be)$ for some $\al, \be \in R$. Suppose that $\mbox{wt}(\al-\be, \al u, \al+\be)=1$. Firstly, if $\al-\be \neq 0$, then $\al u=\al+\be=0$. Since $u$ is not a zero divisor, $\al=0$. But then $\al+\be=0$ implies that $\be=0$. So, $\al-\be=0$, a contradiction. Secondly, if $\al u\neq 0$, then $\al-\be=\al+\be=0$. So, $2\al=0$. Since $2$ is not a zero divisor, $\al=0$. So, $\al u=0$, a contradiction. Thirdly, if $\al+\be\neq 0$, then $\al-\be=\al u =0$. Since $u$ is not a zero divisor, $\al=0$. But then $\al-\be=0$ implies that $\be=0$. So, $\al+\be =0$, a contradiction. So, there is no $\al, \be\in R$ such that $\mbox{wt}(\al-\be, \al u, \al+\be)=1$. Since $(-1, 0,1)\in C_{R_2}$ and $\d_2 \geq 2$, it must follow that $\d_2 =2$.
\end{proof}

\begin{cor}\label{sss}
Let $R$ be such that $2=2.1_R$ is not a zero divisor. If there exist self-orthogonal linear codes $C_1, C_2$ of length $m$ over $R$ with respective minimum distances $d_1, d_2$, then there \linebreak exists a self-orthogonal matrix-product code of length $3m$ over $R$ with minimum distance satisfying $d\geq \mbox{min}\{3d_1, 2d_2\}$.
\end{cor}

\begin{proof}
Using the matrix $A$ of Lemma \ref{diag1}, it follows from Theorem \ref{self-orth 1} that $[C_1 C_2]\,A$ is self-orthogonal. Moreover, by (1), $d\geq \mbox{min}\{3d_1, 2d_2\}$.
\end{proof}

\begin{lem}\label{adiag1}
Let $R$ be such that there exists $u\in R$ with $u^2=-1$.
\begin{itemize}
\item[1.] For $A=\left(\begin{array}{ccc} 1&0&u\\0&1&u\end{array}\right)$, $AA^t=\mbox{adiag}(-1,-1)$ and $\d_1=\d_2=2$.

\item[2.] If further $2=2\cdot 1_R$ is not a zero divisor and $B=\left(\begin{array}{ccccc} 1&u&0&1&u\\u&1&u&0&1\end{array}\right)$, then $BB^t=\mbox{adiag}(3u,3u)$, $\d_1=4$, and $\d_2=3$.
\end{itemize}
\end{lem}

\begin{proof}
Similar to the proof of Lemma \ref{diag1}
\end{proof}

\begin{cor}\label{adiag2}
Let $R$ be such that $-1$ is a perfect square. If there exist self-orthogonal linear codes $C_1, C_2$ of length $m$ over $R$ whose respective minimum distances are $d_1, d_2$ with $C_1\subseteq C_2^\perp$ and $C_2\subseteq C_1^\perp$, then

1. There exists a self-orthogonal matrix-product code of length $3m$ over $R$ with minimum distance satisfying $d\geq \mbox{min}\{2d_1, 2d_2\}$.

2. If further $2\cdot 1_R$ is not a zero divisor, then there exists a self-orthogonal matrix-product code of length $5m$ over $R$ with minimum distance satisfying $d\geq \mbox{min}\{4d_1, 3d_2\}$.
\end{cor}

\begin{proof}
Using the matrices $A$ and $B$ of Lemma \ref{adiag1}, it follows form Theorem \ref{self-orth 2} that $[C_1 C_2]\,A$ and $[C_1 C_2]\,B$ are self-orthogonal of lengths $3m$ and $5m$ and minimum distances satisfying $d\geq \mbox{min}\{2d_1, 2d_2\}$ and $d\geq \mbox{min}\{4d_1, 3d_2\}$, respectively.
\end{proof}

\begin{example}
It is a known fact that if $p$ and $q$ are odd primes, then $-1$ is a perfect square modulo $pq$ if and only if $-1$ is a perfect square modulo each of $p$ and $q$ (see \cite{NZM}). It is a also known that if $p$ is congruent to 1 modulo 4, then $-1$ is a perfect square modulo $p$. Let $p$ be a prime congruent to 1 modulo 4 and $R=\Z_{p^2}$. Then $-1$ is a perfect square in $R$. Let $x=(1,1, \dots, 1)\in R^p$, $y=(p,p, \dots, p)\in R^p$, $C_1= Rx$, and $C_2= Ry$. Then, $d_1=d_2=p$, $C_1 \subseteq C_2^\perp$, and $C_2\subseteq C_1^\perp$. Using the matrices $A$ and $B$ of Lemma \ref{adiag1}, it follows from Corollary \ref{adiag2} that the matrix-product codes $[C_1 C_2]\,A$ and $[C_1 C_2]\,B$ are both self-orthogonal of lengths $3p$ and $5p$ and minimum distances satisfying $d \geq 2p$ and $d\geq 3p$, respectively.
\end{example}

\begin{lem}\label{adiag3}
Let $R$ be such that $2=2.1_R\in U(R)$ and there exists $u\in R$ with $u^2=-1$. Then for $A=\left(\begin{array}{cc} 1&u\\u&1\end{array}\right)$, $AA^t=\mbox{adiag}(2u,2u)$, $\d_1=2$, and $\d_2=1$.
\end{lem}

\begin{proof}
Similar to the proof of Lemma \ref{diag1}
\end{proof}

\begin{cor}\label{adiag4}
Let $R$ be such that $2=2.1_R\in U(R)$ and $-1$ is a perfect square. If there exist linear codes $C_1, C_2$ of length $m$ over $R$ whose respective minimum distances are $d_1, d_2$ with $C_1= C_2^\perp$ and $C_2= C_1^\perp$, then there exists a self-dual matrix-product code of length $2m$ over $R$ with minimum distance satisfying $d\geq \mbox{min}\{2d_1, d_2\}$.
\end{cor}

\begin{proof}
Using the matrix $A$ of Lemma \ref{adiag3}, it follows from Theorem \ref{self-dual} that $[C_1 C_2]\,A$ is self-dual. Moreover, by (1), $d\geq \mbox{min}\{2d_1, d_2\}$.
\end{proof}

\begin{remark}
Under the same assumptions on $R$ of Lemma \ref{adiag3}, a square matrix of any size, like the one in Lemma \ref{adiag3}, can be constructed. If $s$ is even, then
$$A=\left(\begin{array}{cccccccc} 1&0&\dots&0&0&\dots&0&u \\ \vdots&\vdots&\dots&\vdots&\vdots&\dots&\vdots&\vdots \\ 0&0&\dots& 1&u&\dots&0&0 \\ 0&0&\dots&u&1&\dots&0&0\\ \vdots&\vdots&\dots&\vdots&\vdots&\dots&\vdots&\vdots \\ u&0&\dots&0&0&\dots &0&1 \end{array}\right)\in M_{s\times s}(R)$$
satisfies $AA^t=\mbox{adiag}(2u, 2u, \dots, 2u)$, $\d_1=\dots=\d_{s/2}=2$, and $\d_{s/2 +1}=\dots=\d_{s}=1$; while if $s$ is odd, then
$$A=\left(\begin{array}{ccccccccccc} 1&0&\dots&0&0&0&0&0&\dots&0&u \\ \vdots&\vdots&\dots&\vdots&\vdots&\vdots&\vdots&\vdots&\dots&\vdots&\vdots \\ 0&0&\dots&0& 1&0&u&0&\dots&0&0 \\ 0&0&\dots&0&0&1&0&0&\dots&0&0\\ 0&0&\dots&0&u&0&1&0&\dots&0&0\\ \vdots&\vdots&\dots&\vdots&\vdots&\vdots&\vdots&\vdots&\dots&\vdots&\vdots \\ u&0&\dots&0&0&0&0&0&\dots &0&1 \end{array}\right)\in M_{s\times s}(R)$$
satisfies $AA^t=\mbox{adiag}(2u, \dots, 2u, 1, 2u, \dots, 2u)$, $\d_1=\dots=\d_{(s-1)/2}=2$, and $\d_{(s+1)/2}=\dots=\d_s=1$. So, like Corollary \ref{adiag4}, Theorem \ref{self-dual} can be applied once there exist linear codes $C_1, \dots, C_s$ of length $m$ over $R$ whose respective minimum distances are $d_1,\dots, d_s$ with $C_i =C_{s-i+1}^\perp$, for $i=1, \dots, s$, to get a self-dual matrix-product code of length $sm$ and minimum distance satisfying 
\begin{equation*}
d\geq \left \{ \begin{array} {l@{\quad;\quad}l}
\vspace{.3cm}
\mbox{min}\{2d_1,\dots, 2d_{s/2},d_{s/2 +1},\dots,d_{s}\} & \mbox{if $s$ is even} \\
\mbox{min}\{2d_1,\dots, 2d_{(s-1)/2},d_{(s+1)/2 +1},\dots,d_{s}\} & \mbox{if $s$ is odd}.
\end{array} \right.
\end{equation*}

\end{remark}

\section*{Tables}

Finally, the two tables below give concrete examples highlighting the results of this section. All input codes $C_1$ and $C_2$ below are self-dual (and, hence, self-orthogonal), which can be found in references \cite{DKL}, \cite{KL}, or \cite{LDL}. The element -1 in the rings chosen is always a perfect square (see \cite[Lemma 4.2]{DKL} and \cite[Lemma 3.1]{KL}). Recall also that 2 is a unit in any commutative ring of odd characteristic. Generalizing a well-known result, it was shown in \cite{BD} that if $R$ is any commutative ring with identity, $A\in M_{s\times l}(R)$ is of full rank, and $C_1, \dots, C_s$ are free linear codes over $R$ of ranks $k_i$ for $i=1, \dots, s$, then the MPC $[C_1 \dots C_s]\,A$ is free of rank $\sum_{i=1}^s k_i$.

Table 1 concerns self-orthogonal MPCs and Table 2 concerns self-dual MPCs.\\

\begin{center}
\begin{tabular}{|c| c| c| c| c|}
\hline
 $R$ & $C_1$ & $C_2$ & $[C_1 C_2]\,A$ & Reason\\
\hline
$\mbox{GR}(11^2, 2)$ & $[12,6,6]$ & $[12,6,7]$ & $[36,12,d\geq 14]$ & Corollary \ref{sss} \\
\hline
$\mbox{GR}(11^2, 2)$, $\mbox{GR}(5^3, 2)$& $[12,6,6]$ & $[12,6,6]$ & $[36,12,d\geq 12]$ & Corollary \ref{adiag2}(1) \\
\hline
$\mbox{GR}(11^2, 2)$, $\mbox{GR}(5^3, 2)$& $[12,6,6]$ & $[12,6,6]$ & $[60,12,d\geq 18]$ & Corollary \ref{adiag2}(2) \\
\hline
$\mbox{GR}(5^3, 2)$, $\mbox{GR}(3^4, 2)$, $\mbox{GR}(3^2, 2)$ & $[10,5,5]$ & $[10,5,5]$ & $[30,10,d\geq 10]$ & Corollary \ref{adiag2}(1) \\
\hline
$\mbox{GR}(5^3, 2)$ & $[10,5,5]$ & $[10,5,5]$ & $[50,10,d\geq 15]$ & Corollary \ref{adiag2}(2) \\
\hline
$\mbox{GR}(3^2, 2)[x]/(x^2-3)$& $[8,4,5]$ & $[8,4,5]$ & $[24,8,d\geq 10]$ & Corollary \ref{adiag2}(1) \\
\hline
$\Z_{25}$, $\mbox{GR}(3^2, 2)$, $\mbox{GR}(3^2, 2)[x]/(x^2-3)$ & $[6,3,4]$ & $[6,3,4]$ & $[18,6,d\geq 8]$ & Corollary \ref{adiag2}(1) \\
\hline
$\Z_{25}$ & $[6,3,4]$ & $[6,3,4]$ & $[30,6,d\geq 12]$ & Corollary \ref{adiag2}(2) \\
\hline
$\mbox{GR}(3^2, 2)$& $[4,2,3]$ & $[4,2,3]$ & $[12,4,d\geq 6]$ & Corollary \ref{adiag2}(1) \\
\hline
\end{tabular}
\end{center}
\begin{center} \tablename{$\;$1}: Self-orthogonal MPCs \end{center}
\hfill

\begin{center}
\begin{tabular}{|c| c| c| c| c|}
\hline
 $R$ & $C_1$ & $C_2$ & $[C_1 C_2]\,A$ & Reason\\
\hline
$\mbox{GR}(11^2, 2)$ & $[12,6,7]$ & $[12,6,7]$ & $[24,12,d\geq 7]$ & Corollary \ref{adiag4} \\
\hline
$\mbox{GR}(11^2, 2)$, $\mbox{GR}(5^3, 2)$& $[12,6,6]$ & $[12,6,6]$ & $[24,12,d\geq 6]$ & Corollary \ref{adiag4} \\
\hline
$\mbox{GR}(5^3, 2)$, $\mbox{GR}(3^4, 2)$, $\mbox{GR}(3^2, 2)$ & $[10,5,5]$ & $[10,5,5]$ & $[20,10,d\geq 5]$ & Corollary \ref{adiag4} \\
\hline
$\mbox{GR}(3^2, 2)[x]/(x^2-3)$& $[8,4,5]$ & $[8,4,5]$ & $[16,8,d\geq 5]$ & Corollary \ref{adiag4} \\
\hline
$\Z_{25}$, $\mbox{GR}(3^2, 2)$, $\mbox{GR}(3^2, 2)[x]/(x^2-3)$ & $[6,3,4]$ & $[6,3,4]$ & $[12,6,d\geq 4]$ & Corollary \ref{adiag4} \\
\hline
$\mbox{GR}(3^2, 2)$& $[4,2,3]$ & $[4,2,3]$ & $[8,4,d\geq 3]$ & Corollary \ref{adiag4} \\
\hline
\end{tabular}
\end{center}
\begin{center} \tablename{$\;$2}: Self-dual MPCs \end{center}


\section*{Acknowledgement}
The authors extend their thanks to Patrick Sol\'{e} for some useful discussion. They also would like to express their gratitude to King Khalid University for providing administrative and technical support. A. Deajim would also like to thank the University Council and the Scientific Council of King Khalid University for approving a sabbatical leave request for the academic year 2018-2019, during which this paper was prepared.

\end{document}